\documentclass{article}[10pt]
\usepackage{amssymb,slashed,framed,cancel}

\usepackage{graphicx}
\usepackage{amsmath,mathrsfs}

\newtheorem{theorem}{Theorem}

\newtheorem{condition}[theorem]{Condition}

\newtheorem{proposition}[theorem]{Proposition}

\newenvironment{proof}[1][Proof]{\textbf{#1.} }{\ \rule{0.5em}{0.5em}}

\begin{document}

\title{Is Born-Jordan really the universal Path Integral Quantization Rule?}

\author{John E.~Gough \\
\texttt{jug@aber.ac.uk}\\
   Aberystwyth University, SY23 3BZ, Wales, United Kingdom}

\date{}

\maketitle

\begin{abstract}
    It has been argued that the Feynman path integral formalism leads to a quantization rule, and that the Born-Jordan rule is the unique quantization rule consistent with the correct short-time propagator behavior of the propagator for non-relativistic systems. We examine this short-time approximation and conclude, contrary to prevailing views, that the asymptotic expansion applies only to Hamiltonian functions that are at most quadratic in the momentum and with constant mass. While the Born-Jordan rule suggests the appropriate quantization of functions in this class, there are other rules which give the same answer, most notably the Weyl quantization scheme.
\end{abstract}

\section{Introduction}

A quantization rule is a map $\mathscr{Q}$ taking functions $f=f(q,p)$ on phase space with canonical variables $(q,p)$ to operators $\hat f= \mathscr{Q}(f)$ expressed in terms of the canonical observables $\hat q , \hat p$ satisfying $[\hat q , \hat p] = i \hbar \, \hat I$. We list some basic requirements: linearity, i.e., $\mathscr{Q}(af+bg) =a\mathscr{Q}(f)+b \mathscr{Q}(g)$ for $a$ and $b$ constant; that $\mathscr{Q}(f_0)=f_0 \, \hat I$ for $f_0$ constant; and that $\mathscr{Q}(f(q))=f(\hat q)$ and $\mathscr{Q}(g(p))=g(\hat p )$.

An example is the $\tau$- rule $\mathscr{S}_\tau$ \cite{Berezin_Shubin}, for $\tau \in [0,1]$, by giving the kernel in the position representation as
\begin{eqnarray}
\langle q_1 | \mathscr{S}_\tau (f) | q_2 \rangle = 
\int f \big(  (1-\tau ) q_1 + \tau q_2 , p ) e^{i (q_1 -q_2) p/\hbar } \, \frac{dp}{2 \pi \hbar } .
\end{eqnarray}
In particular, one has
\begin{eqnarray}
    \mathscr{S}_\tau (q^s  p^r)
= \sum_{m\ge 0}  \binom{r}{m} (1 -\tau )^m \tau^{r-m} \, \hat p^ {r-m} \hat q^s \hat p^m ,
\label{eq:tau_form}
\end{eqnarray}
The Weyl quantization rule $\mathscr{W}$ is the special case $\mathscr{S}_\tau$ with $\tau = 1/2$. The Weyl quantization rule may be characterized as
\begin{eqnarray}
    e^{i(qp_0 - pq_0)/\hbar } \mapsto e^{i(\hat q p_0 - \hat pq_0)/\hbar } .
\end{eqnarray}
We say that other rules belong to the Cohen class \cite{Cohen1966} if they are characterized by
\begin{eqnarray}
    e^{i(qp_0 - pq_0)/\hbar } \mapsto \Xi (q_0, p_0) \,e^{i(\hat q p_0 - \hat pq_0)/\hbar } .
\end{eqnarray}
The multiplier should satisfy $\Xi(q_0 , 0 )=\Xi (0, p_0) =1$ so that so that functions of $q$ (respectively $p$) only get mapped to their operator counterparts. Ideally, we have that $\Xi$ is non zero so we may apply the rule to general functions via Fourier transform on phase space. The $\tau$-quantization rule corresponds to $\Xi_\tau (q_0,p_0) = e^{-i (\tau - \frac{1}{2}) q_0p_0/\hbar}$.

Quantization rules form a convex set. In particular, let $\mathbb{P}$ be a probability measure on a measurable set $\Omega$ and $\omega \mapsto\mathscr{Q}_\omega (\cdot )$ a measurable assignment of quantization mappings, then $\int_\Omega \mathscr{Q}_\omega (\cdot ) \, \mathbb{P} [ d \omega ]$ is again a quantization. In particular, if $\mathscr{Q}_\omega$ lies in the Cohen class with multiplier $\Xi_\omega$, then so does $\int_\Omega \mathscr{Q}_\omega (\cdot ) \, \mathbb{P} [ d \omega ]$
and its multiplier is then $\int_\Omega \Xi_\omega (\cdot ) \, \mathbb{P} [ d \omega ]$.

Historically, the first quantization rule was the one introduced by Born and Jordan \cite{BJ1925a,BJ1925b}: $\mathscr{B}(q^n p^m) =\frac{1}{n+1}\sum_{m=0}^n \hat q^m \hat p^r \hat q^{n-m}$. The rationale was that this is the unique quantization for which $\frac{1}{i \hbar}[\mathscr{B}(f), \mathscr{B}(g)] = \mathscr{B}\big( \{f,g\} \big)$ whenever $f=f(q)$, $g=g(p)$ and $\{f,g\}$ are Poisson brackets. We also note that the Born-Jordan rule is the uniform randomization of the $\tau$-quantization rule:
\begin{eqnarray}
    \mathscr{B} (\cdot ) \equiv \int_0^1 \mathscr{S}_\tau (\cdot ) \, d \tau .
\end{eqnarray}
For reference, $\mathscr{W}(q^2 p^2)=\frac{1}{6} (\hat q^2 \hat p^2 + \hat q \hat p \hat q \hat p + \hat q \hat p^2 \hat p + \hat p \hat q^2 \hat p +  \hat p \hat q \hat p \hat q+\hat p^2 \hat q^2 )$ while $\mathscr{B}(q^2 p^2) = \frac{1}{3} (\hat q^2 \hat p^2+ \hat q \hat p^2 \hat q +\hat p^2 \hat q^2 )$

\subsection{Feynman's Path Integral Approach}
In non-relativstic quantum mechanics, the state $|\psi \rangle$ of a system evolves from event $A$ with spacetime coordinates $(q_A,t_A)$ to a later event $B$ with $(q_B,t_B)$ by a unitary $\hat U (t_B,t_A) = e^{-i \hat H (t_B-t_A) /\hbar}$ where $\hat H$ is an observable (the Hamiltonian). We may use the Dirac notation $\langle A | \psi \rangle =\psi (A) = \psi (q_A, t_A) $ for the wave-function. The wave-function will then satisfy the Schr\"{o}dinger equation, $i \hbar \frac{\partial }{\partial t} \psi = \hat H \, \psi$.

Once the Hamiltonian $\hat H$ governing the evolution has been fixed, we may introduce the propagator $\langle B|A \rangle =\langle q_B.t_B |q_A,t_A \rangle$, fro $t_B \ge t_A$ which is the Green's function for the Schr\"{o}dinger equation:
\begin{eqnarray}
    \psi (q_B,t_B ) = \int_{\text{space}} 
    \langle q_B , t_B |q_A,t_A \rangle \, \psi (q_A,t_A ) \, dq_A ,
\end{eqnarray}
which may be succinctly written as $\langle B | \psi \rangle =  \int_{\text{space}} 
    \langle B |A \rangle \, \langle A | \psi \rangle \, dq_A $. The propagator itself is\begin{eqnarray}
    \langle B |A \rangle =
    \langle q_B | \hat U (t_B,t_A ) |q_B \rangle 
    \end{eqnarray}
Since the Schr\"{o}dinger equation is first order in time, we get the consistency condition
\begin{eqnarray}
      \langle C |A \rangle =\int_{\text{space}} 
    \langle C | B \rangle \,\langle B |A \rangle\, dq_B 
\label{eq:CK}
\end{eqnarray}
whenever $t_C \ge t_B\ge t_A$.

The consistency condition may be repeated to give
\begin{eqnarray}
    \langle B |A \rangle =\int_{\text{space}^n} 
    \langle B | X_n \rangle \,\langle X_n |X_n-1  \rangle \cdots \langle X_1 | A \rangle \, dq_n \cdots dq_1
\end{eqnarray}
where $X_k$ are events $(q_k, t_k)$ with the chronological ordering $t_B > t_n > t_n-1 >\cdots >t_1> t_A$. We may think of $A,X_1,\cdots, X_n,B$ as events on a path from $A$ to $B$ is spacetime. It was Feynman \cite{Feynman_Hibbs} who took this a step further and 
introduced the path integral expression
\begin{eqnarray}
    \langle B |A \rangle =\int_{B \leftarrow A}
    e^{i S [ \mathsf{q}]/\hbar } \, \mathcal{D} \mathsf{q} 
\end{eqnarray}
Where the integral is over all spacetime paths $\mathsf{q}$ of the form $q=q(t)$ for $t_A \le t \le t_B$ with $q(t_A)=q_A$ and $q(t_B)=q_B$.
The phase term involves the action $S[\mathsf{q}]= \int_{t_A}^{t_B} L\left( q(t), \dot q(t) \right) \, dt$ where the $L$ is the Lagrangian, assumed to be the Legendre transformation $  L(q,v) = \sup_{p} \left( pv - H(q,p) \right)$.

The propagator may also be represented as the phase space path integral
\begin{eqnarray}
    \langle B|A \rangle = \int_{B \leftarrow A} e^{i S[\mathsf{q,p}]/\hbar} \, \mathcal{D} \mathsf{q}
    \mathcal{D} \mathsf{p}
\end{eqnarray}
where we integrate over all paths $\mathsf{q,p}$ in phase space parametrized by time $t_A \le t\le t_B$ with the endpoint conditions $\mathsf{q} (t_A) =q_A$, $\mathsf{q} (t_B) = q_B$.
Here,
\begin{eqnarray}
    S[ \mathsf{q,p}] = \int_{t_A}^{t_B}
    \bigg( p(t) dq(t) -H \left( q(t), p(t) \right) \,dt \bigg).
\end{eqnarray}

If $ t_B-t_A$ is small, one might argue that dominant contribution comes from the linear path (see Figure \ref{fig:kerner_path} below)
\begin{eqnarray}
    q_{\text{lin.}} (t) =q_A + \frac{q_B-q_A}{t_B-t_A} (t-t_A ),
    \quad 
    p_{\text{lin.}} (t) = p \quad \text{(constant)}.
\end{eqnarray}

\begin{figure}[h]
    \centering
    \includegraphics[width=0.5\linewidth]{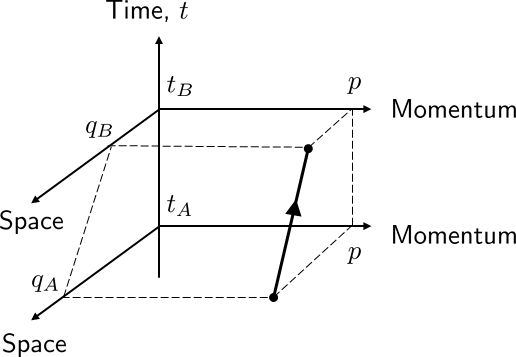}
    \caption{The linear path from $(q_A,p,t_A)$ to $(q_B,p,t_B)$ drawn  in space-momentum-time.}
    \label{fig:kerner_path}
\end{figure}
We then get the small-time approximation \cite{Garrod} (anticipating the correct coefficient)
\begin{eqnarray}
    \langle B|A \rangle \approx \frac{1}{(2 \pi \hbar )} \int_{-\infty}^\infty
    e^{i \left(  p(q_B-q_A ) - \bar{H} (q_A,q_B,p) \, (t_B-t_A ) \right) / \hbar} \, dp,
    \label{eq:zzz}
\end{eqnarray}
where we encounter the average Hamiltonian
\begin{eqnarray}
    \bar H (q_A,q_B,p) &=& \frac{1}{t_B-t_A} \int_{t_A}^{t_B}
    H ( \mathsf{q}_{\text{lin.}} (t) , p ) \, dt 
    \nonumber \\
    &=& \int_0^1 H \left( q_A (1- \tau ) + q_B \tau , p \right) \, d\tau .
    \label{eq:average_H}
\end{eqnarray}

Taking $q_B-q_A$ to be fixed, but $t_B-t_A$ to be small, the exponential in (\ref{eq:zzz}) may be expanded to first order as
\begin{eqnarray}
   \langle B|A \rangle \approx \delta(q_B-q_A ) + \frac{t_B-t_A}{i \hbar} 
    \frac{1}{(2 \pi \hbar )} \int_{-\infty}^\infty
    e^{i  p(q_B-q_A )}  \bar {H} (q_A,q_B,p) \, dp .
    \label{eq:kerner_short}
\end{eqnarray}
The small time behavior of the Schr\"{o}dinger equation tells us that $\psi (B) \approx \psi (q_B,t_A) + \frac{ t_B-t_A}{i \hbar} \left( \hat H \psi \right) (B)$ and by comparison Kerner and Sutcliffe advanced an argument \cite{Kerner} observed that
\begin{eqnarray}
    \langle q_A | \hat H | q_B \rangle \equiv \frac{1}{(2 \pi \hbar )} \int_{-\infty}^\infty
    e^{i  p(q_B-q_A )}  \bar{H} (q_A,q_B,p) \, dp \equiv \langle q_A | \,\mathscr{B} (\hat H ) | q_B \rangle ,
    \label{eq:kerner_int}
\end{eqnarray}
and so they advanced the argument that the Born-Jordan rule emerges as the correct physical rule.

\subsection{Quantization Through Path Integration?}
The Feynman path integral approach \cite{Feynman_Hibbs} apparently gives us a method to go from a classical Hamilton's function $H=H(q,p)$ on phase space to a quantum mechanical model having a well defined Hamiltonian operator $\hat H$. This is achieved by realizing the quantum propagator as a functional integral of a phase terms involving by the action associated to paths, see also the development by Garrod \cite{Garrod}. In principle, this amounts to a quantization procedure!

By the short-time propagator argument above, Kerner and Sutcliffe \cite{Kerner} propose that the privileged quantization rule arising from Feynman's approach is the Born-Jordan rule. However, this was almost immediately challenged by Cohen \cite{Cohen1970} (in fact, both papers appeared in the same issue of Journal of Math. Phys.!) who argued that the limit form for the propagator, when performing the usual Trotter time-slicing \cite{Trotter,Nelson}, was insensitive to the approximation step used for the short-time propagator and that the different choices lead to different quantization rules.

For example, the $\tau$-quantization corresponds to taking the average as an intermediary point evaluation $H \big( (1-\tau )q_A +\tau q_B, p \big)$ along the path such as for $0 \le \tau \le 1$. (Weyl then being the mid-point rule!) Indeed, the only requirement according to Cohen was that the average $\bar H (q_A , q_B , p)$ should converge to $H(q,p)$ as $q_A$ and $q_B$ converge to $q$.

  For Hamiltonians of the form $H= \frac{1}{2m}p^2+V(q)$, the essence of the argument is that the terms is the $N$-time slice form
\begin{eqnarray}
    \langle q_B | \, e^{-i \hat H (t_B-t_A)/\hbar}\, | q_A \rangle =
    \int dq_{N-1} \cdots \int dq_1
\, \prod_{k=1} \langle q_{k+1} | e^{-i \hat H (t_B-t_A)/N\hbar}\, | q_k \rangle
\end{eqnarray}
can be approximated as
\begin{eqnarray}
    \langle q_{k+1} |  \,  e^{-i \hat H \Delta t/\hbar}\, | q_k \rangle
    \approx \sqrt{\frac{m}{2\pi i \, \hbar\Delta t}}
    \exp \bigg\{ \frac{im}{2 \hbar \Delta t} (\Delta q )^2
    -\frac{i \Delta t}{\hbar} \overline{V}(q_k,q_{k+1}) \bigg\}.
    \label{eq:slice}
\end{eqnarray}
Here $N$ is assumed to be large, $\Delta t = (t_B-t_A)/N$ is assumed small, $\Delta q = q_{k+1}-q_k$ is arbitrary, and $\overline{V}(q_k,q_{k+1})$ is an average of the potential between the points $q_k$ and $q_{k+1}$. Roughly speaking, so long as both $\Delta t \to 0^+$ \textit{and} $\Delta q \to 0$ then the choice of average is largely unimportant as far as the limit is concerned: the reason that $\Delta q$ can be made limiting small being due to a stationary phase argument for the $(\Delta q)^2$ term in the exponential in (\ref{eq:slice}). Cohen's rebuttal was uncontested and widely accepted. 

However, the situation was revisited by Makri and Miller \cite{Makri_Miller_88} who calculated the short time limit for the classical action $S_{\text{cl.} }(B|A)$ with $\Delta q$ held fixed. They showed that the correct short-time approximation was the one suggested by Kerner and Sutcliffe. 
Their derivation used a semi-classical approximation \cite{Schulman} and an expansion of the Hamilton-Jacobi equation. Makri and Miller \cite{Makri_Miller_88} arrive at the following approximation
\begin{eqnarray}
    \langle B | A\rangle \approx \sqrt{\frac{m}{2\pi i \, \hbar\Delta t}}
    e^{ i S_{\text{cl.} }(B|A)/\hbar}+ O\big( (\Delta t)^2 \big).
\end{eqnarray}
Note that this is established for Hamilton's functions of the form $H(q,p) = \frac{1}{2m} p^2 +V(q)$ only.

In a somewhat polemic article, Kauffmann \cite{Kauffmann} argued that Cohen's invocation of continuity of paths to argue that $\Delta q$ should vanish as $\Delta t$ goes to zero is irrelevant to this approximation step. We shall refer to this as the \textit{Kauffmann Trap}. Kauffmann \cite{Kauffmann}, and later De Gosson \cite{dG2016,dG2018}, thereby argued that the Born-Jordan quantization rule is the natural one to use physically.

\subsection{Our Contribution}
We shall consider an approximation scheme for phase space trajectories for classical systems governed by a Hamiltonian function $H$, under the assumption that we have a long way to go ($\Delta q$ fixed) and a short time to get there $(\varepsilon=\Delta t \to 0^+)$. We shall refer to this as the \textit{secular} expansion for the phase trajectories. 

Contrary to the claim by de Gosson (Theorem 2 in \cite{dG2016}) that the short-time expansion is valid for general $H$, we find that the secular expansion is valid only for Hamiltonians of the standard for $H(q,p) = \frac{1}{2m}p^2 +A p+ V(q)$ where $m$ and $A$ must be a constant. In the proof of Theorem 2 in \cite{dG2016}, de Gosson assumes that $\Delta q$ (and $\Delta p$) must tend to zero as $\Delta t \to 0^+$ thereby falling into the Kauffmann trap. Of course, $\Delta q$ should be held fixed (following Kauffmann) while $\Delta p$, in fact, actually needs to be of order $1/\Delta t$ for the particle to travel the fixed distance in the short time (these features will be built into the secular approximation).

We then recover the short-time asymptotic expansion for the classical action found by Makri and Miller \cite{Makri_Miller_88}: their derivation follows an alternative method which uses an asymptotic expansion of the semi-classical action \cite{Schulman}. We additionally give a convenient form for the fifth-order term.

\section{The Short-Time Behavior}
In the following, we shall consider classical mechanical systems governed by a Hamiltonian $H=H(q,p)$ which, in turn, can be obtained from a Lagrangian $L=L(q, \dot q)$ by an invertible Legendre transform $L(q,\dot q ) =\max_p \{p\dot q - H(q,p) \}$. With an arbitrary differentiable path $\mathsf{q} = \{ q(t) : t \in [t_A,t_B ] \}$ we may associate the action 
\begin{equation}
    S[\mathsf{q}] = \int_{t_A}^{t_B} L \big(
    q(t), \dot q (t) \big) \, dt .
\end{equation}
Let us impose the endpoint conditions
\begin{eqnarray}
    q(t_A)=q_A, \quad q(t_B) = q_B
\end{eqnarray}
so that the path starts at event $A$ corresponding to $(q_A, t_A)$ and finishes at the later event $B$ given by $(q_B, t_B )$.
Hamilton's Principle tells us that the classical path $\mathsf{q}_{\text{cl.}}$ from $A$ to $B$ is the one that minimizes the action.
Hamilton's principal function (or classical action) is the action $S_{\text{cl.}} (B|A)$ as calculated along the classical path from $A$ to $B$.

Alternatively, let $(\mathsf{q},\mathsf{p})$ be the solution to Hamilton's equations of motion $\dot q = \frac{\partial}{\partial p} H, \dot p = - \frac{\partial}{\partial q} H$ with $\mathsf{q}$ again satisfying the endpoint conditions: then $\mathsf{q}$ is again the classical path from $A$ to $B$. The conjugate path $\mathsf{p}$ in momentum space is fixed by the endpoint conditions. Note that the path $(\mathsf{q},\mathsf{p})$ in phase space is equivalently fixed by giving the data $q_A=q(t_A)$ and $p_A = p(t_A)$ at time $t_A$.

\subsection{The Secular Approximation}
Our aim is to examine the behavior of the classical paths in the short time regime where $\varepsilon = t_B-t_A$ becomes small while the separation $q_B -q_A$ remains finite. To this end, it is convenient to introduce the new parameter
\begin{eqnarray}
    \tau = \frac{t-t_A}{t_B-t_A} \in [0,1].
\end{eqnarray}
The phase trajectory may then be re-parameterized as
\begin{eqnarray}
    \tilde q (\tau ) = q(t) , \qquad \tilde p (\tau )= p(t) .
\end{eqnarray}
For instance, the linear path in configuration space from $A$ to $B$ is defined as
\begin{eqnarray}
q_{\text{lin.}} (t) = q_A + \frac{q_B-q_B}{t_B-t_A}
(t-t_B)
\end{eqnarray}
and this is re-parameterized as $\tilde q_{\text{lin.}} (\tau ) = \overleftrightarrow{q}_{AB} (\tau )$ where it is convenient to introduce the interpolating function
\begin{eqnarray}
    \overleftrightarrow{q} (\tau )= (1-\tau ) q_A + \tau q_B ,
\end{eqnarray}
(we shall assume that $A$ and $B$ are fixed and drop the $AB$ dependence for convenience).
Note that $u_{BA} = \frac{q_B-q_B}{t_B-t_A}$ is the (constant) velocity needed to get from $A$ to $B$ but that this will be of order $1/\varepsilon$ in our regime. We also have $\frac{d}{d \tau}
\overleftrightarrow{q} (\tau )= q_B-q_A$.

In the following, we shall now write $\tilde q_\varepsilon, \, \tilde p_\varepsilon$ to emphasize the $\varepsilon$-dependence of the solutions!

As an example, let us consider the free particle $H= \frac{p^2}{2m}$. The Hamilton's equations of motion imply that $p(t) =p_A$ (constant) and $\dot q = p/m$ leading to the solution $q= q_{\text{lin.}}$. We must, of course, have $p_A = m \frac{q_B-q_A}{t_B-t_A}$. We therefore have $\tilde{q}_\varepsilon \equiv \tilde q_{\text{lin.}}$ and $\tilde{p}_\varepsilon = m (q_B-q_A) \varepsilon^{-1}$. As is well known, the classical action for a free particle is
\begin{eqnarray}
    S_{\text{cl.}}^{\text{free}} (B|A) =
    \frac{ m(q_B - q_A)^2}{2 (t_B-t_A)}=
    \frac{ 1}{2 }m(q_B - q_A)^2\, \varepsilon^{-1}
\end{eqnarray}
and, in particular, diverges as $\varepsilon \to 0^+$. This $O(\varepsilon^{-1} ) $ behavior is, as we shall see, typical.

The classical path in phase space is an integral curve to the Hamiltonian vector field $(u,w)$ where
\begin{eqnarray}
    u(q,p) = \frac{\partial}{\partial p} H (q,p), \qquad w(q,p) =- \frac{\partial}{\partial q} H (q,p).
\end{eqnarray}

\begin{condition}
    The Hamiltonian vector field $(u,w)$ is analytic in $p$:
    \begin{eqnarray}
        u(q,p ) = \sum_{n \ge 0} u_n (q) \, p^n , \qquad
        w(q,p ) = \sum_{n \ge 0} w_n (q) \, p^n ,
    \end{eqnarray}
    where $u_n (q) , w_n (q)$ are smooth. The coefficient function $w_1$ is also assumed to be positive and bounded away from zero. We shall write it as $u_1 (q) = 1/m(q)$ where $m(q)$ may be interpreted as a position-dependent mass. The function $u_0$ has dimensions of velocity.
\end{condition}

The position variable then satisfies $q(t)=q_A +\int_{t_A}^t u \left( q(s),p(s) \right) \, ds$ or, equivalently,
\begin{eqnarray}
\tilde{q} (\tau ) = q_A +\varepsilon \int_0^\tau u\left( \tilde{q} (\sigma ) , \tilde {p} (\sigma )\right) \, d \sigma.
\end{eqnarray}
The end point conditions imply $q_B - q_A =\varepsilon \int_0^1 v\left( \tilde{q} (\tau ) , \tilde {p} (\tau )\right) \, d \tau$ and, since the left hand side is independent of $\varepsilon$, we conclude that
\begin{eqnarray}
    \lim_{\varepsilon \to 0^+}
    \varepsilon \int_0^1 u\left( \tilde{q}_\varepsilon (\tau ) , \tilde {p}_\varepsilon (\tau )\right) \, d \tau = q_B - q_A .
    \label{eq:lim1}
\end{eqnarray}

It is clear that at least one of thepaths $\tilde q_\varepsilon$ or $\tilde p_\varepsilon$ must become divergent in order to compensate for the factor $\varepsilon$ in (\ref{eq:lim1}). However, $\tilde q_\varepsilon$ is just a re-parametrization of a bounded function, so it is $\tilde p_\varepsilon$ which becomes singular in the $\varepsilon \to 0^+$ limit. (This, of course, corresponds to the intuitive picture that the particle needs a large velocity/momentum to travel the fixed distance in a short time!)

In fact, we shall make the following assumption.

\begin{condition}
    We have $\tilde q_\varepsilon  = \overleftrightarrow{q} + O(\varepsilon )$
uniformly on the interval $[0,1]$ and that we have the expansion
\begin{eqnarray}
    \tilde q_\varepsilon (\tau ) = \overleftrightarrow{q} (\tau ) + \sum_{n=1}^\infty \varepsilon^n \, \chi_n (\tau ) .
\end{eqnarray}
where the functions $\chi_n$ are sufficiently smooth.
\end{condition}

Note that the endpoint conditions require
\begin{eqnarray}
    \chi_n (0) = \chi_n (1) = 0, \qquad (n=1,2,3, \cdots ).
    \label{eq:chi_end}
\end{eqnarray}
It follows that the limit (\ref{eq:lim1}) can be recast as
\begin{eqnarray}
    \lim_{\varepsilon \to 0^+}
    \varepsilon \int_0^1 u\left( \overleftrightarrow{q} (\tau ) , \tilde {p}_\varepsilon (\tau )\right) \, d \tau = q_B - q_A .
    \label{eq:lim2}
\end{eqnarray}

From Conditions 1 and 2, we see that $\varepsilon \, v (\tilde q, \tilde p_\varepsilon )\equiv \varepsilon\, u_0 (\tilde q ) + \varepsilon\, u_1 (\tilde q ) \, \tilde p_\varepsilon+ \varepsilon \, u_2 ( \tilde q ) \, \tilde p_\varepsilon^2 +\cdots$ and we require this to have a finite non-zero limit as $\varepsilon \to 0^+$.

\begin{condition}
    The momentum path $\tilde p _\varepsilon$ takes the \textbf{secular} form
    \begin{eqnarray}
        \tilde p_\varepsilon (\tau) = \sum_{n=-1}^\infty \pi_n (\tau ) \, \varepsilon^n = \pi_{-1} (\tau ) \varepsilon^{-1} + \pi_0 (\tau )+ \pi_1 (\tau ) \varepsilon+ \pi_2 (\tau) \varepsilon^2 + \cdots.
        \label{eq:secular}
    \end{eqnarray}
\end{condition}
Note that the leading term is taken to $O(\varepsilon^{-1} )$ in line with the singular behavior already seen for the free particle.

\subsection{Restrictions on the Hamiltonian}
The assumption that the solutions are secular places severe restrictions on the choice of Hamiltonian.

\begin{proposition}
    The secular condition applies only for Hamiltonians that are at most quadratic in momentum.
\end{proposition} 

\begin{proof}
From $\dot q = u(q,p)$ we find $\frac{d}{d \tau} \tilde q_\varepsilon  = \varepsilon \, u ( \tilde q_\varepsilon , \tilde p _\varepsilon ) \equiv \varepsilon\sum_n u(\tilde q_\varepsilon ) \, \tilde p_\varepsilon^n$. As $\tilde p_\varepsilon = O(\varepsilon^{-1})$ we see that the $n$th term will diverges as $\varepsilon \to 0^+$ if $n \ge 2$, therefore we must truncate to get $u (q,p) =u_0 (q) +u_1(q) \, p \equiv \frac{\partial}{\partial p} H(q,p)$ and so
    \begin{eqnarray}
        H(q,p) \equiv \frac{1}{2 } u_1 (q) p^ 2
        + u_0 (q) p +V(q) .
        \label{eq:quad_Ham}
    \end{eqnarray}
\end{proof}

We therefore restrict to Hamiltonians of the form $H(q,p) \equiv \frac{1}{2 m(q)} \, p^ 2 + u_0 (q) \, p +V(q) $. Our next order of business is to investigate the asymptotic behavior of the canonical equations of motion themselves. 

\begin{proposition}
    The secular expansion is valid if and only if the mass $m$ is constant. We also find that
    \begin{eqnarray}
        \pi_{-1} = m (q_B-q_A).
    \end{eqnarray}
\end{proposition}

\begin{proof}
    Hamilton's $q$-equation reads as $\frac{d}{d \tau} \tilde q_\varepsilon  = \varepsilon \, \frac{1}{m} \tilde p _\varepsilon$. Inserting the expansions, we see that the $O(1)$ and $O(\varepsilon)$ terms yield
    \begin{eqnarray}
        q_B-q_A &=& \frac{1}{m} \, \pi_{-1} (\tau ),
        \label{eq:piminus} \\
        \frac{d}{d\tau} \chi_1 (\tau ) &=& \frac{1}{m}\, \pi_0 (\tau ),
        \label{eq:dchi1}
    \end{eqnarray}
    respectively. Hamilton's $p$-equation reads as $\frac{d}{d \tau} \tilde p_\varepsilon  = - \varepsilon u_0'(\tilde q_\varepsilon) \, \tilde p_\varepsilon- \varepsilon \,  V' ( \tilde q_\varepsilon )$ and again inserting the secular expansion we find that the $O(\varepsilon^{-1} )$ and $O(1)$ terms are
    \begin{eqnarray}
        \frac{d}{d \tau} \pi_{-1} (\tau ) =0 ,
        \label{eq:dpiminus}\\
        \frac{d}{d \tau} \pi_0 (\tau)= - u_0' (\overleftrightarrow{q}(\tau ))\, \pi_{-1}.
        \label{eq:dpi0}
    \end{eqnarray}
    From (\ref{eq:piminus}), we see that $\pi_{-1}= m (q_B-q_A)$ and from (\ref{eq:dpiminus}) this must be constant. As a result $m$ is constant. 
    \end{proof}

We see that we are restricted to the general class of Hamiltonian's describing a particle of constant mass in an electromagnetic potential. The term $u_0$ will carry dimensions of speed and would naturally arise from a vector potential contribution in the presence of a magnetic field. For clarity, we will drop this term for the remainder of this section and just consider Hamiltonians of the traditional form $H(q,p) = \frac{1}{2m}p^2+ V(q)$.

\begin{proposition}
    For $H(q,p) = \frac{1}{2m}p^2+ V(q)$. Then $\pi_0=0$ and $\chi_1=0$.
\end{proposition}
\begin{proof}
    It now follows from (\ref{eq:dpi0}) that $\frac{d}{d\tau} \pi_0 =0$ so $\pi_0$ is constant. But since $\int_0^1 pi_0\, d \tau =0$ we deduce that the constant value is zero.

    Likewise, from (\ref{eq:dchi1}) we now have $\frac{d}{d\tau}\chi_1 \equiv 0$. Therefore, $\chi_1$ is also constant, and by the endpoint conditions (\ref{eq:chi_end}).
\end{proof}

\begin{proposition}
    Taking the Hamiltonian now to be $H(q,p) = \frac{1}{2m}p^2+ V(q)$, we have that
    \begin{eqnarray}
        \frac{d}{d\tau} \chi_n (\tau ) = \frac{1}{m}\pi_{n-1}, \qquad (n=1,2,3, \cdots ).
        \label{eq:chi_pi}
    \end{eqnarray}
    This implies $\int_0^1 \pi_n (\tau ) \, d\tau =0$ and, in particular, that $\pi_2, \pi_4 \equiv 0$ and $\chi_3 \equiv 0$. 
\end{proposition}

\begin{proof}
    The position Hamilton's equation yields
    \begin{eqnarray}
        \frac{d}{dt}\tilde q_\varepsilon &=& \varepsilon\, \frac{1}{m} \, \tilde p_\varepsilon \equiv (q_B-q_A) +\varepsilon\, \frac{1}{m} \pi_1 
        + \varepsilon^2 \, \frac{1}{m}\pi_2  + \cdots 
    \end{eqnarray}
    and, if we compare this to $\frac{d}{dt}\tilde q_\varepsilon \equiv (q_B-q_A)  + \varepsilon^2 \frac{d}{dt}\chi_2 +  \varepsilon^3 \frac{d}{dt}\chi_3 +\cdots$, we obtain the relations (\ref{eq:chi_pi}). The property that $\int_0^1 \pi_n (\tau ) \, d\tau =0$ follows from the same argument used in the $n=0$ case.

    Likewise, the momentum Hamilton's equation yields
    \begin{eqnarray*}
        \frac{d \tilde p_\varepsilon}{d\tau}  &=& -\varepsilon \, V' ( \tilde q_\varepsilon ) 
        \\
        &\equiv& -\varepsilon \, V' ( \overleftrightarrow{q}  )
        -\varepsilon^3 \, V'' ( \overleftrightarrow{q}  )\chi_2    -\varepsilon ^4  V'' ( \overleftrightarrow{q}  ) \chi_3 + O(\varepsilon^5).
    \end{eqnarray*}
    We also have $\frac{d}{dt}\tilde p_\varepsilon \equiv \frac{1}{\varepsilon} \, \frac{d}{d\tau}\pi_{-1} + \frac{d}{d\tau } \pi_0  + \varepsilon \frac{d}{dt}\pi_1 + \cdots$ (the first and second term are now known to vanish) and comparing the coefficients of the powers of $\varepsilon$ yields
    \begin{eqnarray}
        \frac{d}{d\tau} \pi_1 &=&- \, V' ( \overleftrightarrow{q}  ) ,\\
        \frac{d}{d\tau} \pi_2 &=& 0,  \\
        \frac{d}{d\tau} \pi_3 &=& -V'' ( \overleftrightarrow{q}  )\chi_2  ,\\
        \frac{d}{d\tau} \pi_4 &=& - V'' ( \overleftrightarrow{q}  ) \chi_3 .
    \end{eqnarray}
    The development becomes more complicated beyond this point as the Taylor expansion still involves powers of sequences in $\varepsilon$. From the fact that $\pi_2$ is clearly constant, and its average over [0,1] vanishes, it follows that $\pi_2 = 0$. As $\frac{d}{d\tau}\chi_3= \pi_2/m \equiv 0$ with $\chi_3(0)=0$, we conclude that $\chi_3 \equiv 0$ as well. This in turn implies that $\frac{d }{d \tau} \pi_4=0$ which forces $\pi_4=0$.
\end{proof}

\subsection{Specific Examples}
\subsubsection{Particle Experiencing a Constant Force}
We take the Hamiltonian to be $H^F(q,p) = \frac{p^2}{2m}-qF$ in which case the exact solution to Hamilton's equations is
\begin{eqnarray}
    q^F(t)=q_A+\frac{p_A}{m} (t-t_A) + \frac{F}{2m} (t-t_A)^2, \quad p^F(t) =p_A +F (t-t_A) ,
\end{eqnarray}
in terms of the initial data $(q_A,p_A)$. The solution may be re-expressed in terms of the endpoint conditions and here the datum $p_A$ should be replaced by $p_A= m(q_B-q_A)/(t_B-t_A)- \frac{1}{2}F (t_B-t_A)$. Rendering the solution in terms of the parameter $\tau$ then yields
\begin{eqnarray}
    \tilde q _\varepsilon^F (\tau ) = \overleftrightarrow{q}(\tau ) - \frac{1}{2}\frac{F}{m} (\tau -\tau^2 ) \varepsilon^2, \,
    \tilde p^F (\tau ) = m (q_B-q_A) \varepsilon^{-1}
    +F(\tau - \frac{1}{2} \tau ) \, \varepsilon .
\end{eqnarray}
Condition 2 is satisfied with $\chi_2^F (\tau ) = - \frac{1}{2}\frac{F}{m} (\tau -\tau^2 )  $ and $\chi_n \equiv 0$ otherwise. We also have condition 3 satisfied with
\begin{eqnarray}
    \pi_{-1}^F = m (q_B-q_A) , \quad\pi_{1} ^F(\tau ) = F(\tau - \frac{1}{2}  ) , \quad \pi_n^F \equiv 0 \,\, \,  \mathrm {otherwise}.
\end{eqnarray}

The classical action for this problem can be calculated exactly and is given by
\begin{eqnarray}
    S_{\text{cl.}}^F (B|A) &=&  S_{\text{cl.}}^{\mathrm{free}} (B|A)
    + \frac{1}{2} F (q_A+q_B) (t_B-t_A ) - \frac{1}{24} \frac{F^2}{m} (t_B-t_A)^3 \nonumber \\
    &=& \frac{1}{2}m (q_B-q_A)^2 \, \varepsilon^{-1} 
    +\frac{1}{2} F (q_A+q_B) \, \varepsilon - \frac{1}{24} \frac{F^2}{m} \, \varepsilon^3.
\end{eqnarray}

\subsubsection{The Harmonic Oscillator}
The oscillator Hamiltonian is $H^{\text{osc.}}(q,p) = \frac{p^2}{2m}+ \frac{1}{2}m\omega^2 q^2$ and again the exact solution is well known.
Rendered in terms of the parameter $\tau$, this will be
\begin{eqnarray}
    \tilde q^{\text{osc.}} _\varepsilon (\tau )
    &=& q_A \cos ( \omega \tau \varepsilon  ) + \frac{q_B-q_A \cos ( \omega \varepsilon)}{ \sin ( \omega \varepsilon )}
    \sin ( \omega \tau \varepsilon), \nonumber \\
    \tilde p ^{\text{osc.}}(\tau ) &=& - m\omega ^2 q_A \sin ( \omega \tau \varepsilon  ) + m \omega\frac{q_B-q_A \cos ( \omega \varepsilon)}{ \sin ( \omega \varepsilon )}
    \cos ( \omega \tau \varepsilon)
        .
\end{eqnarray}
The asymptotic behavior for small $\varepsilon$ is
\begin{eqnarray}
    \tilde q^{\text{osc.}} _\varepsilon (\tau )
    &=& \overleftrightarrow{q}(\tau ) + \chi_2 ^{\text{osc.}} (\tau ) \, \varepsilon^2
    + \chi_4 ^{\text{osc.}} (\tau ) \, \varepsilon^4 + \cdots ,
    \nonumber \\
    \tilde p ^{\text{osc.}}(\tau ) &=& m (q_B-q_A ) \, \varepsilon^{-1} + \pi_1 ^{\text{osc.}}(\tau ) \varepsilon + \pi_3 ^{\text{osc.}}(\tau ) \varepsilon ^3 + \cdots.
\end{eqnarray}
We see that Conditions 2 and 3 are satisfied. Only the even terms $\chi^{\text{osc.}}_n $ and odd terms $\pi_n ^{\text{osc.}}$ are non-zero with the lowest terms being
\begin{eqnarray}
    \chi^{\text{osc.}}_2 (\tau ) &=& \frac{1}{2} \omega^2 q_A \, \tau (1-\tau^2) 
    + \frac{1}{6} \omega^2 q_B \tau (1-\tau^3) ,
    \nonumber \\
    \pi^{\text{osc.}}_1 (\tau ) &=& 
    m \omega^2 \left( \frac{2q_A +q_B}{6} - q_A \tau - \frac{1}{2} ( q_B -q_A) \tau^2 \right).
\end{eqnarray}

The classical action for the oscillator is also known exactly and is given by
\begin{eqnarray*}
    S_{\text{cl.}}^{\text{osc.}} (B|A) &=&  \frac{m \omega}{2 \sin (\omega (t_B-t_A))}
    \bigg\{ (q_A^2 +q_B^2) \cos (\omega (t_B-t_A)) -2 q_Aq_B\bigg\}.
\end{eqnarray*}
This admits the asymptotic expansion
\begin{eqnarray}
    S_{\text{cl.}}^{\text{osc.}} (B|A) &=&
\frac{m}{2} (q_B-q_A)^2 \, \varepsilon^{-1} 
    -\frac{m \omega^2 }{6} (q_A^2 +q_B^2 +q_Aq_B ) \, \varepsilon + O(  \varepsilon^3).
\end{eqnarray}


\subsection{Short-Time Asymptotics of the Classical Action}
We now give the general case where the secular expansion is valid. Let us take this opportunity to collect all our observations together so far. We deal with Hamiltonians of the form $H= \frac{1}{2m} p^2 + V (q)$ and expansions
\begin{eqnarray}
    \tilde q_\varepsilon &=& \overleftrightarrow{q}+ \varepsilon^2 \chi_2 + \varepsilon^4 \chi_4 + \varepsilon^5 \chi_5 + \cdots ,\\
    \tilde p _ \varepsilon &=& \frac{1}{\varepsilon}\pi_{-1} + \varepsilon \pi_1 + \varepsilon^3 \pi_3 + \varepsilon^5 \pi_5 + \cdots ,
\end{eqnarray}
with $\pi_{-1}= m (q_B-q_A)$, $d \chi_{n+1}= \frac{1}{m}\pi_n \, d \tau$, $ \frac{d}{d\tau}\pi_1 = - V'( \overleftrightarrow{q})$ and $ \frac{d}{d\tau}\pi_3 = - V''( \overleftrightarrow{q})$.

In the following, we will use the notation
\begin{eqnarray}
    \overline{f} \triangleq 
    \int_0^1 f \big( \overleftrightarrow{q} (\tau )\big) \, d \tau =
    \int_0^1 f \big( (1-\tau ) q_A +\tau q_B \big) \, d\tau.
\end{eqnarray}

\begin{theorem}
    Under Conditions 1,2, and 3, the classical action associated with Hamiltonian $H(q,p) = \frac{1}{2m}p^2+ V(q)$ has the short time asymptotic behavior
    \begin{eqnarray}
    S_{\text{cl.}}(B|A) &=& \varepsilon^{-1} 
         \frac{1}{2} m (q_B-q_A)^2 
         - \varepsilon\,\overline{V} 
         -\varepsilon ^3 \frac{1}{2m(q_B-q_A)^2}    \bigg( \overline{V^2}-\overline{V}^2 \bigg) 
         \nonumber \\
         &&   +\varepsilon ^5 \frac{1}{m(q_B-q_A)^2}    \bigg( \overline{VF}-\overline{V}\, \, \overline{F} \bigg) 
         +O({\varepsilon^6}) ,
         \label{eq:S_expand}
    \end{eqnarray}
    where $F = - V'$ is the conservative force derived from the potential $V$.
\end{theorem}

\begin{proof}
    The classical action may be written as 
    \begin{eqnarray}
        S_{\text{cl.}} (B|A) = \int_{t_A}^{t_B} \big( pdq - H(q,p)dt\big)
        =  \int_0^1 \big( \tilde p_\varepsilon d\tilde q_\varepsilon -\varepsilon\,  H (\tilde q_\varepsilon , \tilde p _\varepsilon ) \, d\tau \big)
    \end{eqnarray}
    and we split this into two parts. The first part $ \int_0^1 \tilde p_\varepsilon (\tau ) d \tilde q_\varepsilon (\tau )$ is
    \begin{eqnarray*}
       \int_0^1 \bigg(\frac{\pi_{-1}}{\varepsilon} +  \pi_1 \varepsilon + \pi_3 \varepsilon^3 +\cdots \bigg)
        \bigg( (q_B-q_A)  \, d\tau+ \varepsilon^2 d \chi_2 + \varepsilon^4 d \chi_4  \cdots \bigg)
    \end{eqnarray*}
but we encounter terms of the form $(q_B-q_A)\pi_{-1} \int_0^1\pi_n d \tau$ and $\pi_{-1}\int_0^1 d\chi_n$ which vanish. This leaves
\begin{eqnarray}
     \int_0^1 \tilde p_\varepsilon (\tau ) d \tilde q_\varepsilon (\tau ) = \frac{1}{\varepsilon}\frac{(q_B-q_A)^2}{m}+\varepsilon^3\int_0^1 \pi_1 d \chi_2 +
     \varepsilon^5 \int_0^1 \big( \pi_3 d \chi_2+ \pi_1 d\chi_4 \big) + O(\varepsilon^6).
\end{eqnarray}

The second part is
\begin{eqnarray*}
         \varepsilon\int_0^1 H ( \tilde q_\varepsilon , \tilde p_\varepsilon ) \, d\tau & = & \varepsilon\int_0^1
        \bigg\{ \frac{1}{2m}  \tilde p_\varepsilon^2 +V (\overleftrightarrow{q} + \varepsilon^2 \chi_2 + \cdots) \bigg\} d \tau
        \nonumber \\
        &=& \varepsilon^{-1} 
         \frac{1}{2} m (q_B-q_A)^2 
         + \varepsilon\int_0^1  V (\overleftrightarrow{q} ) \, d\tau 
         \nonumber \\
         && \qquad +\varepsilon ^3 \int_0^1 \bigg( \frac{\pi_1^2}{2m} 
         + V' (\overleftrightarrow{q} ) \chi_2 (\tau ) \bigg) d\tau\nonumber \\
         && \qquad +\varepsilon ^5 \int_0^1  \frac{\pi_1\pi_3}{m} d\tau
          +O(\varepsilon^6)  .
\end{eqnarray*}

Combining the two terms leads to
\begin{eqnarray}
    S_{\text{cl.}}(B|A) &=& \varepsilon^{-1} 
         \frac{1}{2} m (q_B-q_A)^2 
         - \varepsilon\int_0^1 V (\overleftrightarrow{q} ) \, d\tau 
         \nonumber \\
         && \qquad +\varepsilon ^3 \int_0^1 \bigg( 
         \pi_1 d\chi_2 -\frac{\pi_1^2}{2m} d\tau
         -  V' (\overleftrightarrow{q} ) \chi_2 (\tau )  \, d\tau 
        \bigg)\nonumber \\
         && \qquad +\varepsilon ^5 \int_0^1 \bigg( 
         \pi_1 d \chi_4 +\pi_3 d \chi_2 -\frac{\pi_1\pi_3}{m} d\tau \bigg)
         +O({\varepsilon^6}) .
\end{eqnarray}

For the $O(\varepsilon^3)$ term, we observe that
$\int_0^1 \pi_1 \, d \chi_2 =\int_0^1 \frac{\pi_1^2}{m}d\tau$ and 
\begin{eqnarray*}
    \int_0^1 \big( -V' (\overleftrightarrow{q} ) \big) \,  \chi_2 d\tau \ =\int_0^1 \big( \frac{d\pi_1}{d\tau} \big) \,  \chi_2\, d \tau = -\int_0^1 d \chi_2\,\pi_1  = - \int_0^1 \frac{\pi_1^2}{m}d\tau.
\end{eqnarray*} 
The full order 3 contribution is therefore $- \varepsilon^3\int_0^1 \frac{\pi_1^2}{2m}d\tau.$

Starting from $\frac{d}{d\tau} \pi_1 = - V' (\overleftrightarrow{q}) $, we have
\begin{eqnarray*}
    \pi_1 (\tau ) &=&-\int_0^\tau V' \big( (1-\sigma ) q_A + \sigma q_B \big) \, d \sigma + \pi_1 (0) \\
    &=&- \frac{1}{(q_B-q_A )} \bigg( V \big( (1- \tau ) q_A + \tau q_B \big) - V_A\bigg) + \pi_1 (0).
\end{eqnarray*}
From $ \int_0^1 \pi_1 (\tau ) \, d \tau -0$, we see that $\pi_1(0) = (q_B-q_A) [\overline{V} -V(q_A)]$ and therefore
\begin{eqnarray}
    \pi_1 (\tau ) = -\frac{1}{(q_B-q_A)}\bigg( V \big( (1- \tau ) q_A + \tau q_B \big) - \overline{V} \bigg).
\end{eqnarray}
We also note that $\pi_1(1) =  -(q_B-q_A)^{-1}\bigg( V ( q_B) - \overline{V} \bigg) $. 
We therefore have
\begin{eqnarray*}
    \frac{1}{2m}\int_0^1 \pi_1^2 \, d\tau =
    \frac{1}{2m(q_B-q_A)^2} (\overline{V^2}-\overline{V}^2 ).
\end{eqnarray*}

Similarly, for the $O(\varepsilon^5 )$ term we note that $\int_0^1\pi_3\, d \chi_2 = \int_0^1 \pi_1 \, d \chi_4 = \int_0^1 \frac{\pi_1 \pi_3}{m} \, d\tau$. An almost identical argument as before leads to the expression
\begin{eqnarray}
    \pi_3 (\tau ) = -\frac{1}{(q_B-q_A)}\bigg( V' \big( (1- \tau ) q_A + \tau q_B \big) - \overline{V'} \bigg).
\end{eqnarray}
This leads to the form stated in the Theorem.
\end{proof}

\bigskip

Our expansion (\ref{eq:S_expand}) agrees with the the one obtained by Makri and Miller \cite{Makri_Miller_88}, though we are also able to give the $O(\varepsilon^5)$ explicitly.

The $O(\varepsilon )$ contribution involves the uniform average of the potential over the linear trajectory between the endpoints $q_A$ to $q_B$. The $O(\varepsilon^3)$ term then involves the variance of the potential, while the $O(\varepsilon^5)$ involves the covariance of the potential and the force it derives.
\section{Magnetic Hamiltonians}
We will now give the expansion when we include the magnetic term $u_0$. The Hamiltonian is now $H= \frac{1}{2m}p^2+u_0(q)p+V(q)$ and we find
$\frac{d}{d\tau} \tilde q_\varepsilon = \frac{\varepsilon}{m}\tilde p_\varepsilon+ \varepsilon \, u_0 (\tilde q_\varepsilon )$ and $\frac{d}{d\tau} \tilde p_\varepsilon = -\varepsilon u_0' \big( \tilde q_\varepsilon \big) \tilde p_\varepsilon - \varepsilon \, V' (\tilde q_\varepsilon )$. We will see that terms such as $\pi_0$ and $\pi_2$ no longer vanish, however, $\chi_1$ and $\chi_3$ still do.

We begin with the lowest order terms in the expansion and here we obtain
\begin{eqnarray}
    \pi_{-1} = m (q_B-q_A) , \qquad \frac{d\chi_1}{d \tau}= \frac{\pi_0}{m}+ u_0 (\overleftrightarrow{q}), \qquad 
    \frac{d\pi_0}{d\tau} = - u_0' (\overleftrightarrow{q})\, \pi_{-1}.
\end{eqnarray}
By inspection, we see that $\frac{d^2 \chi_1}{d\tau^2}= \frac{d\pi_0}{d\tau}-  (q_B-q_A) u_0' (\overleftrightarrow{q})\equiv 0$, implying that $\chi_1(\tau ) = a\tau +b$. However the end point conditions $\chi_1(0)=\chi_1(1)=0$ again force $\chi_1\equiv 0$. As a consequence we learn that $\pi_0$ is non-zero and given by
\begin{eqnarray}
    \pi_0 =- m u_0 ( \overleftrightarrow{q}).
\end{eqnarray}

The next order terms yield the relations
\begin{eqnarray}
     \frac{d\chi_2}{d \tau}= \frac{\pi_0}{m},  \qquad 
    \frac{d\pi_1}{d\tau} = - u_0' (\overleftrightarrow{q})\, \pi_{0}
    - V' (\overleftrightarrow{q}).
\end{eqnarray}
Here we see that $ \frac{d\pi_1}{d\tau} = - mu_0(\overleftrightarrow{q}) u_0' ( \overleftrightarrow{q})
    - V' (\overleftrightarrow{q})$ which we can integrate to obtain
    \begin{eqnarray}
        \pi_1= \frac{1}{q_B-q_A}\bigg(mu_0(\overleftrightarrow{q})^2 - V (\overleftrightarrow{q}) -m\overline{u_0^2}+ \overline{V} \bigg)
    \end{eqnarray}
where the additive constant ensures that $\int_0^1 \pi_1\, d \tau=0$ (which still follows since $\chi_1 =0$).

The next terms are
\begin{eqnarray}
     \frac{d\chi_3}{d \tau}= \frac{\pi_2}{m} + u_0 '(\overleftrightarrow{q} )\, \chi_2,  \qquad 
    \frac{d\pi_2}{d\tau} = - u_0' (\overleftrightarrow{q})\, \pi_{1}
    - \pi_{-1} u_0'' (\overleftrightarrow{q}).
\end{eqnarray}
Combining these, we find that $\frac{d^2 \chi_3}{d\tau^2}=0$ and, similar to before, we must have $\chi_3 =0$. As a consequence, we find that $\pi_2$ is not zero, but is given by
\begin{eqnarray}
    \pi_2 =- m \, u_0'( \overleftrightarrow{q}) .
\end{eqnarray}

We see that the $O(\varepsilon^{-1}$ term in the short-time development of the classical action $S_{\text{cl.}} (B|A)$ is the same free particle term already encountered. The $O(1)$ term is however
\begin{eqnarray*}
    \bigg( \int_0^1 \cancel{\pi_0 (q_B-q_A )} d\tau \bigg)
    -\bigg( \int_0^1 [ \cancel{\frac{\pi_{-1}\pi_0}{m}}+\pi_{-1} u_0 (\overleftrightarrow{q})]d\tau \bigg)
\end{eqnarray*}
The crossed out terms cancel exactly leaving the final term which integrates to $-m (q_B-q_A) \overline{u_0}$.

The $O(\varepsilon $ term is
\begin{eqnarray*}
    \bigg( \int_0^1 \cancel{\pi_1 (q_B-q_A )} d\tau \bigg)
    -\bigg( \int_0^1 [ \cancel{\frac{\pi_{-1}\pi_1}{m}}+ \frac{\pi_0^2}{2m}+\pi_{0} u_0 (\overleftrightarrow{q})]d\tau \bigg)
\end{eqnarray*}
and this equals $\frac{1}{2} m \overline{u_0^2}-\overline{V}$.

At $O(\varepsilon^2)$ we get the contribution
\begin{eqnarray*}
    \bigg( \int_0^1 \cancel{\pi_0  d\chi_2 }\bigg)
    -\bigg( \int_0^1 [ \cancel{\frac{\pi_{0}\pi_1}{m}}+ \bcancel{\frac{\pi_{-1}\pi_2}{m}}+\pi_{1} u_0 (\overleftrightarrow{q})]d\tau +\bcancel{\pi_{-1} u_0'(\overleftrightarrow{q}) \chi_2}\bigg)
\end{eqnarray*}
which leaves
\begin{eqnarray*}
    -\int_0^1 u_0(\overleftrightarrow{q})\pi_1\, d \tau
    \equiv - \frac{1}{q_B-q_A}\bigg( m \overline{u_0^3} - \overline{u_0 V} - m \overline{u_0}\overline{u_0^2}+ \overline{u_0}\overleftrightarrow{V} \bigg).
\end{eqnarray*}

The short-time expansion is then
 \begin{eqnarray}
    S_{\text{cl.}}(B|A) &=& \varepsilon^{-1} 
         \frac{1}{2} m (q_B-q_A)^2 -m (q_B-q_A) \overline{u_0} \nonumber \\
         && \varepsilon \bigg( \frac{1}{2} m \overline{u_0^2}-\overline{V}\bigg)
         -\varepsilon ^2    \frac{1}{q_B-q_A}\bigg( m \overline{u_0^3} - \overline{u_0 V} - m \overline{u_0}\overline{u_0^2}+ \overline{u_0}\overline{V} \bigg)
         \nonumber \\
         &&    +O({\varepsilon^3}) ,
         \label{eq:Smag_expand}
    \end{eqnarray}
    We see that there are now nonzero $O(1)$ and $O(\varepsilon^2 )$ terms appearing due to the magnetic term $u_0$.
\section{Path Integral Approximations}
We now present an alternative strategy for defining quantizations based on the observations of Cohen \cite{Cohen1970} and of Smolyanov, Tokarev, and Truman \cite{Smol_Truman}.

A quantization rule $\mathscr{Q}$ can be defined as above by averaging over $\tau$-quantization rules with a probability measure $\mathbb{P}$ on [0,1]. Here we assign to each continuous function $H$ on phase space a continuous function $\bar H_{\mathscr{Q}}$ given by
\begin{eqnarray}
    \bar H_{\mathscr{Q}} = \int_0^1 H \big( (1-\tau ) q_A + \tau q_B ,p \big) \, \mathbb{P}[d \tau ]
\end{eqnarray}
with the property that
\begin{eqnarray}
    \lim_{q_B \to q_A} \bar H_{\mathscr{Q}} (q_A, q_B ,p) = H(q_A , p).
\end{eqnarray}

For a given prescription, we define a pseudo-differential operator $\hat H$ by
\begin{eqnarray}
    \left( \hat H \psi \right) (q_B)
    =\lim_{E \to \infty} 
    \int
    \bar H_{\mathscr{Q}}  (q_A, q_B ,p)\vert_{\le E }
    \,
    e^{i p (q_B-q_A)/\hbar} 
    \, \psi ( q_A)
    \, \frac{dq_A dp}{2 \pi \hbar}
\end{eqnarray}
where the limit is in the $L^2( \mathbb{R},dq)$-norm. The domain of $\hat H$ consists of those $\psi \in L^2 (\mathbb{R},dq)$ for which the limit exists. (This generalizes the definition in \cite{Smol_Truman} which treats the $\tau$-quantization case.)

We refer to $H$ as the $\mathscr{Q}$-symbol for the pseudo-differential operator $\hat H$. Ideally, there is a self-adjoint operator $\mathscr{Q} (H)$ for which the domain of $\hat H$ is a core, and which agrees with $\hat H$ on this domain. In this case, we would refer to $\mathscr{Q}(H)$ as the quantization of $H$.

This construction shows that the operator ambiguity arises from the averaging prescription used, as originally suggested by Cohen \cite{Cohen1970}.

Let us suppose that we have a quantization procedure $\mathscr{Q}$ for turning phase space functions into operators. In general, the operator product $\mathscr{Q}(f) \, \mathscr{Q}(g)$ will be something of the form $\mathscr{Q}(f \star g)$ where $\star$ is an (asymmetric though associative) product of phase space functions, necessarily different from the usual point-wise product.

Suppose that we fix a Hamiltonian $H$ on phase space and consider the dynamics governed by $\hat H =\mathscr{Q}(H)$. The unitary evolution operator will be
\begin{eqnarray}
    \hat{U} (t_B,t_A ) =
    e^{-i(t_B-t_A) \mathscr{Q}(H)/\hbar}.
\end{eqnarray}
While we might expect $\mathscr{Q}$ to be linear, given that we have argued that it involves an averaging $H\to \bar H$ step (whichever average we select!) and an integration step (\ref{eq:kerner_int}), however, it will generally not be the case that $\mathscr{Q}(H^n)$ is not $\mathscr{Q}(H)^n$ for $n>1$.

However, we might expect that for $\psi$ in a core domain of $\mathscr{Q}(H)$ we at least have
\begin{eqnarray}
    e^{-it \mathscr{Q}(H)/\hbar}
    \, \psi =
    \lim_{n\to \infty} \left( \mathscr{Q}(e^{-it H/ n\hbar} ) \right)^n \psi ,
    \label{eq:limit}
\end{eqnarray}
with convergence in the Hilbert norm. The limit may be rewritten as
\begin{eqnarray}
    \lim_{n\to \infty}  \mathscr{Q} \left(\underset{n-\text{fold}}{ 
    e^{-it H/ n\hbar}  \star \cdots \star e^{-it H/ n\hbar}   }
    \right) \psi .
\end{eqnarray}

Proving the convergence rigorously can be challenging, but we can describe a method introduced by Smolyanov, Tokarev and Truman \cite{Smol_Truman}. It is based on the following Theorem of Chernoff\cite{Chernoff}.

\begin{theorem}[Chernoff, \cite{Chernoff}]
    Let $:t\mapsto \hat F (t)$ be a strongly continuous function from $[0, \infty)$ taking values in the bounded operators over a fixed Hilbert space. We suppose that $\| \hat F (t) \| \le 1$, $\hat F (0)= \hat I$, and that the operator $\frac{d}{dt} \hat F (0^+)$ is closable with closed extension $\frac{1}{i\hbar}\hat H$ with $\hat H$ self-adjoint. Then $\hat F (\frac{t}{n})^n$ converges to $e^{-i t \hat H/ \hbar}$ in the strong operator topology and uniformly in time for $t\in [0,T]$, for each $T<\infty$.
\end{theorem}

It follows that the convergence in (\ref{eq:limit}) holds if we can establish that $\hat F (t) = \mathscr{Q}(e^{-it H/n \hbar })$ satisfies the conditions of Chernoff's Theorem.

This strategy is implemented in \cite{Smol_Truman} for the $\tau$-quantization in the case where the Hamiltonian takes the form $H=T+V$ where $T=T(p)$ is real-valued and $V_0(q,p)$ is real-valued and is the (phase space) Fourier transform of some measure. There, they show convergence for the class of quantizations $\mathscr{S}_\tau$.


\section{Conclusion}
Hamiltonian mechanics (both classical and quantum!) has the truly remarkable feature that the infinitesimal generators of one-parameter groups of canonical transformations correspond to the physical observables. 

The Feynman path integral, in principle, gives us a route to
replace classical observables (functions on phase space) with quantum observables (self-adjoint operators). Despite initial criticism by Cohen, the proposal by Kerner and Sutcliffe that this should be the Born-Jordan rule found subsequent justification.

As we show in this paper, to avoid falling into Kauffmann's trap, the class of Hamiltonians $\mathcal{O}$ that are amenable to the short-time limit with a fixed spatial separation turns out to be those of the form $\frac{1}{2m}p^2+u_0 (q)p+V(q)$. 
We have worked with a one-dimensional system here for transparency, but this is easily generalized and we see that the typical Hamiltonian we consider in 3 dimensions is just the standard non-relativistic Hamiltonian for a particle of charge $-e$:
\begin{eqnarray}
    H= \frac{1}{2m} | \mathbf{p}+ e \mathbf{A} |^2  -e \Phi .
\end{eqnarray}
with scalar potential $\Phi$ and vector potential $\mathbf{A}$ dependent on the position vector $\mathbf{q}=(x,y,z)$. (We have ignored the issue of time-dependent Hamiltonians so far, however, this presents no major difficulty in the short-time limit and is easily dealt with. Therefore, we can readily accommodate time-dependent potentials as well!)

The program of isolating the Born-Jordan rule as the ``correct'' quantization procedure still has a significant gap. Foregoing the argument of whether the Born-Jordan rule is the unique mechanism that realizes the correct short-time propagator behavior and thereby provides the correct self-adjoint Hamiltonian to corresponding to the classical one, we still have a problem: it applies only to classical observables that lie in the class $\mathcal{O}$!

It then selects the requirement that the quantization rule satisfy the universal  
\begin{eqnarray}
    \mathscr{Q} \big( p f(q) \big) \equiv \mathscr{B}\big( p f(q) \big) .
\end{eqnarray}
It is worth recalling that the Born-Jordan rule $\mathscr{B}$ \cite{BJ1925b} started with the requirement that
\begin{eqnarray}
\mathscr{B} \big( \phi (q) \psi (p) \big) \triangleq \frac{1}{i \hbar} [ \Phi ( \hat q) , \Psi (\hat p ) ],
\label{eq:BJ_rule}
\end{eqnarray}
where $\Phi$ and $\Psi$ are anti-derivatives of $\phi$ and $\psi$, respectively. (That is, $\phi = \Phi'$ and $\psi = \Psi'$.) But in the present context we are restricted to $\Phi (p) = a p^2 +bp$.

It is then easy to see that Born-Jordan rule leads to
\begin{eqnarray}
    \mathscr{B} \big( \phi (q) \,p \big) = \frac{1}{2} \phi (\hat q) \, \hat p + \frac{1}{2}\hat p \phi (\hat  q ) .
\end{eqnarray}
However, this is not the only rule that does this. From (\ref{eq:tau_form}), we have
\begin{eqnarray}
     \mathscr{S}_\tau \big ( \phi (q) \,p \big) =( 1-\tau ) \phi (\hat q) \, \hat p + \tau \, \hat p \phi (\hat  q ) .
\end{eqnarray}
so the Weyl rule also satisfies this property. In fact, every quantization of the form $\int_0^1 \mathscr{S}_\tau (\cdot ) \, \mathbb{P}[d \tau ]$ will satisfy this provided that $\mathbb{P}$ has mean 1/2.

\end{document}